\documentclass[journal]{IEEEtran}

\usepackage[draft]{changes}  
\definechangesauthor[name={Anni}, color=red]{AL}

\usepackage{bm}
\usepackage{amssymb}
\usepackage{epsfig}
\usepackage{amsmath}
\usepackage{amsthm}
\usepackage{xcolor}
\usepackage{subfigure}
\usepackage{graphicx}
\usepackage{epstopdf}
\usepackage{amsfonts}%
\usepackage{indentfirst}
\usepackage[ruled]{algorithm2e}
\newtheorem{problem}{\textbf{Problem}}
\newtheorem{definition}{\textbf{Definition}}

\newtheorem{theorem}{\rm\textbf{Theorem}}

\newtheorem{remark}{\rm\textbf{Remark}}
\newtheorem{proposition}{\rm\textbf{Proposition}}

\begin{document}
%
\title{Taylor-Lagrange Control for Safety-Critical Systems}
%
%
%

\author{Wei Xiao and Anni Li
\thanks{W. Xiao is with the Robotics Engineering Department at Worcester Polytechnic Institute, Worcester, MA, 01609, USA
	\texttt{{\small wxiao3@wpi.edu}}}
\thanks{A. Li is with the Electrical and Computer Engineering Department at The University of North Carolina at Charlotte, Charlotte, NC, 28223, USA
	\texttt{{\small ali20@charlotte.edu}}}
\thanks{Manuscript received December 8, 2025.}}

\maketitle

\begin{abstract}

This paper proposes a novel Taylor-Lagrange Control (TLC) method for nonlinear control systems to ensure the safety and stability through Taylor's theorem with Lagrange remainder. To achieve this, we expand a safety or stability function with respect to time along the system dynamics using the Lie derivative and Taylor's theorem. 
This expansion enables the control input to appear in the Taylor series at an order equivalent to the relative degree of the function. We show that the proposed TLC provides necessary and sufficient conditions for system safety and is applicable to systems and constraints of arbitrary relative degree. The TLC exhibits connections with existing Control Barrier Function (CBF) and Control Lyapunov Function (CLF) methods, and it further extends the CBF and CLF methods to the complex domain, especially for higher order cases. Compared to High-Order CBFs (HOCBFs), TLC is less restrictive as it does not require forward invariance of the intersection of a set of safe sets while HOCBFs do. We employ TLC to reformulate a constrained optimal control problem as a sequence of quadratic programs with a zero-order hold implementation method, and demonstrate the safety of zero-order hold TLC using an event-triggered control method to address inter-sampling effects. Finally, we illustrate the effectiveness of the proposed TLC method through an adaptive cruise control system and a robot control problem, and compare it with existing CBF methods.
\end{abstract}


%
\IEEEpeerreviewmaketitle

\section{Introduction}
\label{sec:intro}
\IEEEPARstart{W}{ith} the rapid growth of artificial intelligence and autonomy, 
the problem of stabilizing dynamical systems while optimizing costs and satisfying safety constraints has received significant attention in recent years. Traditional methods like optimal control \cite{Bryson1969} \cite{kirk2004optimal} and dynamic programming \cite{bellman1966dynamic} \cite{Denardo2003}  are primarily designed to linear systems and constraints, limiting their applicability to safety-critical nonlinear systems. Model Predictive Control (MPC) methods \cite{garcia1989model} \cite{Rawlings2018} have been widely applied to receding horizon control, but they are generally computationally expensive, especially for nonlinear systems. Although linearization or simplification can be adopted to improve MPC computational efficiency, such approximations may compromise safety guarantees. Reachability analysis \cite{althoff2014online} \cite{asarin2003reachability} is also widely used to verify system safety, but it incurs heavy computation load. In order to address the computational challenge, barrier-based methods have received increasing attention for nonlinear systems.

Barrier functions (BFs) are used primarily in optimizations with inequality constraints \cite{Boyd2004} by incorporating the reciprocal form of the constraints into the cost function, and they are also used in learning systems to improve safety, such as in safe reinforcement learning \cite{cheng2019end}. However, this method cannot strictly guarantee system safety as the safety is taken as part of the cost function or reward, and thus its functionality is similar to a soft constraint.  BFs are also used as Lyapunov-like functions \cite{Wieland2007}, and they have been employed in verification and control for set invariance \cite{Tee2009}\cite{Aubin2009}\cite{Prajna2004}\cite{Prajna2007}\cite{Wisniewski2013}. 

\begin{figure}[t]
	\centering
	\includegraphics[scale=0.5]{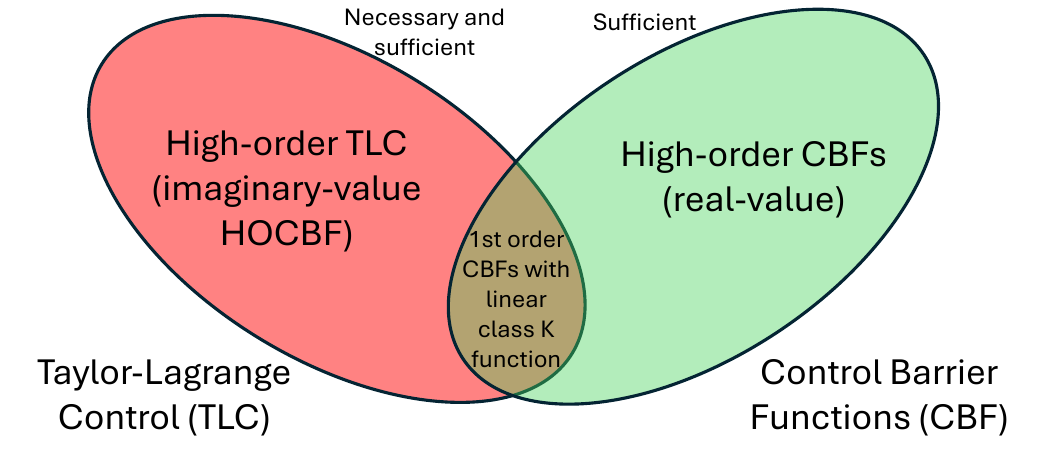}
	\caption{Comparison between the propose TLC method and the CBF method. The proposed TLC and CBF coincide when the relative degree of the safety constraint is one, while the TLC extends HOCBF to the imaginary value domain in high-order cases. The existence of a CBF is only a sufficient condition for system safety, while the TLC is necessary and sufficient.}	
	\label{fig:compare}
\end{figure}

Control Barrier Functions (CBFs) are the extensions of BFs for control systems. CBFs are primarily based on the Nagumo's theorem \cite{Nagumo1942berDL} showing that if the safety function is initially non-negative and the derivative of the safety function is always non-negative at the boundary of safe sets, then the safety function remains non-negative, i.e., the safety can be guaranteed. For CBFs to be used in safe control synthesis, class $\mathcal{K}$ functions are employed in reciprocal CBFs \cite{Ames2017} and zeroing CBFs \cite{Glotfelter2017} to allow the safety function to decrease when the state is far from the boundary. {Both reciprocal CBFs and zeroing CBFs are fundamentally the same, and they are conservative due to the class $\mathcal{K}$ functions. In other words, the existence of a CBF is only a sufficient condition for system safety. There are different variations for CBFs. Stochastic CBFs \cite{clark2021control} have been proposed for stochastic systems, and finite-time convergence \cite{srinivasan2018control} and robust CBFs \cite{cohen2022robust} are for systems that initially violate the safety constraint. Adaptive CBFs \cite{taylor2020adaptive} \cite{xiao2021adaptive} are also introduced to deal with system uncertainty and noise. In order to enforce safety for systems and constraints with high relative degree, exponential CBFs (ECBFs) \cite{Nguyen2016} and high-order CBFs (HOCBFs) \cite{xiao2021high} are also proposed. However, the existence of ECBFs or HOCBFs implies the forward invariance of the intersection of a set of safe sets, which is very restrictive and may potentially limit the system performance.}

In this paper, we propose a novel Taylor-Lagrange Control (TLC) method for system safety and stability using the Taylor's theorem with Lagrange remainder. The TLC method expands a safety function with respect to time and evaluates its time derivative along the system dynamics. When the order of the Taylor expansion equals the relative degree of the safety constraint, we have the control input appears explicitly in the expansion. This enables direct transformation of state constraints into control constraints without using Nagumo's theorem, which is crucial for control synthesis. {The TLC can also be used for system stability in a similar way to system safety.} Finally, we propose to employ TLC to transform constrained optimal control problems into a sequence of Quadratic Programs (QPs) with Zero-Order Hold (ZOH) method, and extend TLC with an event-triggered mechanism to address inter-sampling effects and guarantee continuous-time safety.

As illustrated in Fig. \ref{fig:compare}, unlike CBFs, which provide only sufficient conditions for safety due to class $\mathcal{K}$ functions, TLC provides necessary and sufficient conditions for system safety.
The proposed TLC exhibits connections with CBFs and HOCBFs but provides important extensions. Specifically, TLC pushes the boundary of HOCBFs from the real-value domain to the complex domain, and the equivalent HOCBF form of a TLC may involve complex values in the class $\mathcal{K}$ functions.
Furthermore, while HOCBFs require forward invariance of the intersection of multiple safe sets, TLC requires only forward invariance of the original safe set, making it less restrictive. When applied to stability analysis, TLC similarly relates to Control Lyapunov Functions (CLFs).

The contributions of this paper are as follows:
\begin{enumerate}
    \item We propose a novel Taylor-Lagrange Control method for system safety and stability, and demonstrate its connections with existing CBF, CLF and HOCBF methods under certain conditions.

    \item We provide rigorous proofs of the necessity and sufficiency of TLC for system safety.

    \item We employ TLC to transform constrained optimal control problems into a sequence of QPs, and introduce an event-trigged TLC to address the inter-sampling effect.
\end{enumerate}

The remainder of the paper is organized as follows. In Section \ref{sec:prelim}, we provide preliminaries on Taylor's theorem with Lagrange remainder, CBFs, CLFs and HOCBF. We propose the Taylor-Lagrange certificate for nonlinear systems and Taylor-Lagrange control for control systems in Section \ref{sec:tlc}, and propose Taylor-Lagrange stability for control systems in Sec. \ref{sec:stability}. We formulate optimal control problems that are solved by the Taylor-Lagrange control method
in Section \ref{sec:oc}. An event-triggered Taylor-Lagrange control method is also proposed in Section \ref{sec:oc} to address the inter-sampling effect. The case studies and simulation results are shown in Section \ref{sec:case}. We conclude with final remarks and directions for future work in Section \ref{sec:conclusion}.

\section{PRELIMINARIES}
\label{sec:prelim}

In this section, we introduce the background on the Taylor's theorem with Lagrange remainder, Control Barrier Functions (CBFs), and Control Lyapunov Functions (CLFs).

\subsection{Taylor's theorem with Lagrange remainder}

For completeness and self-contain-ness, we present Taylor's theorem with Lagrange remainder followed by a proof.

\begin{theorem} \label{thm:TL}(Taylor's Theorem with Lagrange Remainder \cite{taylor1717methodus} \cite{Lagrange1797})
Given a function $h:[0,\infty) \rightarrow \mathbb{R}$ that is $m$ times differentiable, $h(t), t\in[0,\infty)$ can be written as the Taylor's expansion with a Lagrange remainder:
\begin{equation}
\begin{aligned}
    h(t) = \sum_{k=0}^{m-1} \frac{h^{(k)}(t_0)}{k!} (t-t_0)^k + \frac{h^{(m)}(\xi)}{m!} (t-t_0)^m, \\t_0\in[0,\infty), \xi\in(t_0, t),
\end{aligned}
\end{equation}
\end{theorem}
\begin{proof}
    Following the fundamental theorem of calculus, we have that:
    \begin{equation} \label{eqn:1}
        h(t) = h(t_0) + \int_{t_0}^t h'(t_1)dt_1.
    \end{equation}

    We can further apply fundamental theorem of calculus to $h'$ in the interval $(t_0,t_1)$:
    \begin{equation}
    \label{eqn:2}
        h'(t_1) = h'(t_0) + \int_{t_0}^{t_1} h''(t_2)dt_2.
    \end{equation}

    Substituting (\ref{eqn:2}) into (\ref{eqn:1}), we have:
    \begin{equation} \label{eqn:4}
    \begin{aligned}
        h(t) &= h(t_0) + \int_{t_0}^t \left(h'(t_0) + \int_{t_0}^{t_1} h''(t_2)dt_2\right)dt_1 \\&= h(t_0) + h'(t_0)(t-t_0)+\int_{t_0}^t\int_{t_0}^{t_1} h''(t_2)dt_2dt_1.
    \end{aligned}
    \end{equation}

    Following a similar procedure, we recursively substitute the following equations into (\ref{eqn:4}):
    \begin{equation}
    \label{eqn:5}
        h^{(k)}(t_k) = h^{(k)}(t_0) + \int_{t_0}^{t_k} h^{(k+1)}(t_{k+1})dt_{k+1}, k\in[2,m-1],
    \end{equation}
    and we get:
    \begin{equation} 
    \label{eqn:6}
\begin{aligned}
    h(t) = \sum_{k=0}^{m-1} \frac{h^{(k)}(t_0)}{k!} (t-t_0)^k + R_m(t)
\end{aligned}
\end{equation}
where
\begin{equation}
\label{eqn:7}
    R_m(t) = \int_{t_0}^t\int_{t_0}^{t_1}\dots \int_{t_0}^{t_{m-1}} h^{(m)}(t_{m})dt_{m}\dots dt_2dt_1.
\end{equation}
Equation (\ref{eqn:6}) is the Taylor's expansion with an integral remainder. 

In (\ref{eqn:7}), we wish to do integration over $t_1$ first, and then $t_2, \dots, t_{m-1}$ since $h^{(m)}(t_{m})$ only depends on $t_m$. Following  Fubini's theorem, equation (\ref{eqn:7}) can be reformulated as
\begin{equation}
\label{eqn:8}
\begin{aligned}
    R_m(t) &= \int_{t_0}^t \int_{t_{m}}^{t}\dots \int_{t_2}^{t} h^{(m)}(t_{m})dt_{1}\dots dt_{m-1}dt_{m} \\&= \int_{t_0}^t h^{(m)}(t_{m})\frac{(t - t_m)^{m-1}}{(m-1)!}dt_{m}.
\end{aligned}
\end{equation}

By the mean value theorem, the integral part in $R_m(t)$ can be replaced by $(t-t_0)$ multiplied by the mean value at some $\xi_1 \in(t_0, t)$, and thus (\ref{eqn:8}) can be reformulated as
\begin{equation}
\label{eqn:9}
\begin{aligned}
    R_m(t) = h^{(m)}(\xi_1)\frac{(t - \xi_1)^{m-1}}{(m-1)!}(t-t_0), \xi_1\in(t_0, t).
\end{aligned}
\end{equation}

Again, by the mean value theorem (or  extreme value theorem and intermediate value theorem), equation (\ref{eqn:9}) can be rewritten as
\begin{equation}
\label{eqn:10}
\begin{aligned}
    R_m(t) = h^{(m)}(\xi)\frac{(t - t_0)^m}{m!}, \xi\in(t_0, t),
\end{aligned}
\end{equation}
where the $\xi$ in the above is different from the $\xi_1$ in (\ref{eqn:9}). This wraps up the proof.

\end{proof}

\subsection{Control Barrier Functions and Control Lyapunov Functions}

Consider an affine control system 
\begin{equation} \label{eqn:affine}%
\dot {\bm{x}} = f(\bm x) + g(\bm x)\bm u,
\end{equation}
where  $\bm x\in X\in \mathbb{R}^n$, $f$ is as defined above, $g:\mathbb{R}^n \rightarrow \mathbb{R}^{n\times q}$ is locally Lipschitz, and $\bm u\in U \subset \mathbb{R}^q$ ($U$ denotes the control constraint set). Solutions $\bm x(t)$ of (\ref{eqn:affine}), starting at $\bm x(0)$, $t\geq 0$, are forward complete.

\begin{definition}
    (Class $\mathcal{K}$ function \cite{Khalil2002}) A function $\alpha:[0, a)\rightarrow[0, \infty), a > 0$ is said to belong to class $\mathcal{K}$ if it is strictly increasing and passes the origin.
\end{definition}

\begin{definition}
    (Forward invariance \cite{Ames2017}) A set $C$ is forward invariant for (\ref{eqn:affine}) if its solution $x(t)\in C, \forall t\geq 0$ for any $x(0)\in C$.
\end{definition}

\begin{definition}
    (Relative degree \cite{Khalil2002}) The relative degree of a function $b(\bm x)$ (or a constraint $b(\bm x)\geq 0$) w.r.t. (\ref{eqn:affine}) is defined as the minimum number of times we need to differentiate $b(\bm x)$ along system (\ref{eqn:affine}) until any control component of $\bm u$ explicitly shows up in the corresponding derivative.
\end{definition}

For a safety constraint $b(\bm x)\geq0$ that has relative
degree $m$ w.r.t system (\ref{eqn:affine}), $b:\mathbb{R}^{n}\rightarrow\mathbb{R}$, we first define a sequence of functions $\psi_{i}:X\rightarrow\mathbb{R},i\in\{1,\dots,m\}$:
\begin{equation}
\begin{aligned} \psi_i(\bm x) := \dot \psi_{i-1}(\bm x) + \alpha_i(\psi_{i-1}(\bm x)),i\in\{1,\dots,m\}, \end{aligned} \label{eqn:functions}%
\end{equation}
where $\psi_{0}(\bm
x):=b(\bm x)$, $\alpha_{i}(\cdot),i\in\{1,\dots,m\}$ denote $(m-i)^{th}$ order
differentiable class $\mathcal{K}$ functions.

We then define a sequence of safe sets $C_{i}, i\in\{1,\dots,m\}$ corresponding to (\ref{eqn:functions}):
\begin{equation}
\label{eqn:sets}\begin{aligned} C_i := \{\bm x \in X: \psi_{i-1}(\bm x) \geq 0\}, i\in\{1,\dots,m\}. \end{aligned}
\end{equation}

\begin{definition}
	\label{def:hocbf} (\textit{High Order Control Barrier Function (HOCBF)}
	\cite{xiao2021high}) Let $C_{i}, i\in \{1,\dots, m\}$ be defined by (\ref{eqn:sets}%
	) and $\psi_{i}(\bm x), i\in\{1,\dots, m\}$ be defined by
	(\ref{eqn:functions}). A function $b: \mathbb{R}^{n}\rightarrow\mathbb{R}$ is
	a High Order Control Barrier Function (HOCBF) of relative degree $m$ for
	system (\ref{eqn:affine}) if there exist $(m-i)^{th}$ order differentiable
	class $\mathcal{K}$ functions $\alpha_{i},i\in\{1,\dots,m-1\}$ and a class
	$\mathcal{K}$ function $\alpha_{m}$ such that 
	
	{\small\begin{equation}
		\label{eqn:constraint}\begin{aligned} \sup_{\bm u\in U}[L_f\psi_{m-1}(\bm x) + L_g\psi_{m-1}(\bm x)\bm u + \alpha_m(\psi_{m-1}(\bm x))] \geq 0, \end{aligned}
		\end{equation}
	}
	\noindent for all $\bm x\in C_{1} \cap,\dots, \cap C_{m}$. In
	(\ref{eqn:constraint}), $L_{f}$ (or $L_{g}$) denotes the Lie derivative along
	$f$ (or $g$). Moreover, it is assumed that $L_g\psi_{m-1}(\bm x) \ne 0$ when $b(\bm x) = 0$.
\end{definition}

The HOCBF is a general form of the relative degree one CBF \cite{Ames2017},
\cite{Glotfelter2017}. In other words, setting $m=1$ reduces a HOCBF to
the common CBF form:
\begin{equation}\label{eqn:cbf0}
L_fb(\bm x) + L_gb(\bm x)\bm u + \alpha_1(b(\bm x))\geq 0.
\end{equation}
We may also add a time-varying penalty function on each of the class $\mathcal{K}$ functions in a HOCBF, and this results in the so-called adaptive CBFs \cite{xiao2021adaptive}.

\begin{theorem} \label{thm:hocbf}
	(\cite{xiao2021high}) Given a HOCBF $b(\bm x)$ from Def. \ref{def:hocbf} with the associated sets $C_1, C_2,\dots, C_{m}$ defined by (\ref{eqn:sets}), if $\bm x(0) \in C_1 \cap C_2\cap,\dots,\cap C_{m}$, then any Lipschitz continuous controller $\bm u(t)$ that satisfies the HOCBF constraint in (\ref{eqn:constraint}), $\forall t\geq 0$ renders
	$C_1$ $\cap  C_2\cap,\dots, \cap C_{m}$ forward invariant for system (\ref{eqn:affine}).
\end{theorem}

\begin{definition}
	\label{def:clf} (\textit{Exponentially stabilizing Control Lyapunov Function (CLF)} \cite{Aaron2012}) A
	continuously differentiable function $V: \mathbb{R}^{n}\rightarrow\mathbb{R}$
	is an exponentially stabilizing control Lyapunov function (CLF) for system
	(\ref{eqn:affine}) if there exist positive constants $c_{1}, c_{2}, c_{3}$ such
	that for $\forall\bm x\in X$, $c_{1}||\bm x||^{2} \leq V(\bm x)
	\leq c_{2} ||\bm x||^{2}, $
	\begin{equation} \label{eqn:clf}
	\inf_{\bm u\in U} \lbrack L_{f}V(\bm x)+L_{g}V(\bm x)
	\bm u + c_{3}V(\bm x)\rbrack\leq0.
	\end{equation}
	
\end{definition}

CBFs/HOCBFs and CLFs are usually employed to transform a nonlinear constrained optimal control problem into a sequence of QPs while ensuring the system safety and stability \cite{Ames2017}.

\section{Taylor-Lagrange Certificate and Control for Safety}
\label{sec:tlc}

In this section, we propose Taylor-Lagrange Certificate (TLCe) to certify the safety of nonlinear systems and Taylor-Lagrange Control (TLC) for obtaining state-feedback safe control policies for nonlinear control systems. 

\subsection{Taylor-Lagrange Certificate}

Consider a system of the form
\begin{equation}\label{eqn:sys}%
\dot {\bm{x}} = f(\bm x),
\end{equation}
with $\bm x\in X \in \mathbb{R}^n$ ($X$ denotes the state constraint set) and $f:\mathbb{R}^n\rightarrow \mathbb{R}^n$ globally Lipschitz. Solutions $\bm x(t)$ of (\ref{eqn:sys}), starting at $\bm x(0)$, $t\geq 0$, are forward complete.

Consider a safety constraint $h(\bm x)\geq 0$ for system (\ref{eqn:sys}) where $h:\mathbb{R}^n \rightarrow \mathbb{R}$ is $r\in\mathbb{N}$ times differentiable. 
Then we define a safe set $C$ based on $h$: 
\begin{equation}\label{eqn:C-1}
C := \{\bm x\in \mathbb{R}^n:h(\bm x)\geq 0\},
\end{equation} 
where $h(x) = 0$ at the boundary of $C$ and $h(x) < 0$ outside $C$.

We expand the safety function $h(\bm x)$ with respect to the time, and define the following Taylor-Lagrange certificate:
\begin{definition} [Taylor-Lagrange Certificate]\label{def:tlce}
    The $r\in\mathbb{N}$ times differentiable function $h:\mathbb{R}^n\rightarrow \mathbb{R}$ is a Taylor-Lagrange Certificate (TLCe) of order $m\in\{1,\dots, r\}$ for system (\ref{eqn:sys}) if
    \begin{equation} \label{eqn:tlce}
       \begin{aligned}
     \sum_{k=0}^{m-1} \frac{L_f^kh(\bm x(t_0))}{k!} (t-t_0)^k + \frac{L_f^mh(\bm x(\xi))}{m!} (t-t_0)^m \geq 0, \\t_0\in[0,\infty), \xi\in(t_0, t),
\end{aligned}
    \end{equation} for all $\bm x(t_0)\in C$,
    where $L_fh$ denotes the Lie derivative of $h$ along $f$.
\end{definition}

\begin{remark} [Interpretation of TLCe]
When the system is initially safe, i.e., $h(\bm x(t_0))\geq 0$, then the TLCe expands the safety constraint $h(\bm x(t))\geq 0$ based on the initial time $t_0$ for some $t > t_0$. In other words, the TLCe evaluates $h(\bm x(t))$ based on the initial condition at $t_0$ and the evolution of the system state within $[t_0, t]$ along the system dynamics (\ref{eqn:sys}).
\end{remark}

\begin{theorem}
\label{thm:tlce}
    Given a set as in (\ref{eqn:C-1}), if there exists a TLCe $b:C\rightarrow \mathbb{R}$, then C is forward invariant for system (\ref{eqn:sys}).
\end{theorem}
\begin{proof}
    Following explicit dynamics analysis \cite{ascher1997implicit}, the evolution of the state $\bm x$ of system (\ref{eqn:sys}) can be represented as an explicit function of time. Therefore, the safety constraint $h(\bm x)\geq 0$ can be rewritten as $h(t)\geq 0, \forall t\geq t_0$ where $h:[0,\infty)\rightarrow \mathbb{R}$.
    
    By Thm. \ref{thm:TL}, we have that 
    \begin{equation}
\begin{aligned}
    h(t) = \sum_{k=0}^{m-1} \frac{h^{(k)}(t_0)}{k!} (t-t_0)^k + \frac{h^{(m)}(\xi)}{m!} (t-t_0)^m \\= \sum_{k=0}^{m-1} \frac{L_f^kh(\bm x(t_0))}{k!} (t-t_0)^k + \frac{L_f^mh(\bm x(\xi))}{m!} (t-t_0)^m, \\t_0\in[0,\infty), \xi\in(t_0, t),
\end{aligned}
\end{equation}

By (\ref{eqn:tlce}), we have that $h(\bm x(t)) = h(t) \geq 0, \forall t\geq t_0$. Therefore, the set $C$ is forward invariant for system (\ref{eqn:sys}).
\end{proof}

The proposed TLCe has some connections with existing Barrier Functions (BFs) and High-Order BFs (HOBFs) \cite{xiao2021high}, and it also pushes their boundaries. The following proposition shows their relationship:

\begin{proposition} \label{prop:tlce}
    Assume $L_f^mh(\bm x(t_0)) \ne 0$ and $sgn(L_f^mh(\bm x(\xi))) = sgn(L_f^mh(\bm x(t_0)))$ for a TLCe in Def. \ref{def:tlce}. The first-order TLCe  degenerates to an adaptive BF with a linear class $\mathcal{K}$ function, and the high-order TLCe may extend the HOBF to the imaginary value domain.
\end{proposition}
\begin{proof}
\textbf{First-order case.} When the TLCe has an order $m = 1$, the corresponding first-order TLCe (\ref{eqn:tlce}) becomes:
\begin{equation} \label{eqn:1tlce}
    h(\bm x(t_0)) + L_fh(\bm x(\xi))(t-t_0)\geq 0, \xi\in(t_0, t).
\end{equation}

Since $L_f h(\bm x(t_0))\ne 0$, the above equation can be rewritten as
\begin{equation} \label{eqn:acbf1}
    h(\bm x(t_0)) + L_fh(\bm x(t_0))p(\xi)(t-t_0)\geq 0, \xi\in(t_0, t), 
\end{equation}
where
\begin{equation}
    p(\xi) = \frac{L_fh(\bm x(\xi))}{L_fh(\bm x(t_0))}.
\end{equation}

As we assume $p(\xi) > 0$ (the signs of $L_fh(\bm x(\xi)), L_fh(\bm x(t_0))$ are assumed to be the same), equation (\ref{eqn:acbf1}) is equivalent to the adaptive barrier function/certificate with a linear class $\mathcal{K}$ function:
\begin{equation} \label{eqn:acbf2}
    \frac{1}{p(\xi)(t-t_0)}h(\bm x(t_0)) + L_fh(\bm x(t_0))\geq 0, \xi\in(t_0, t), 
\end{equation}
where $\frac{1}{p(\xi)(t-t_0)}$ is the adaptive penalty function depending on the future time $\xi\in(t_0, t)$ when we consider $t_0$ as the current time. When varying $t_0$ between $[0,\infty)$, the above equation offers an adaptive barrier function/certificate \cite{xiao2021adaptive}. Note that $L_fh(\bm x(\xi))\ne 0$ in the above equation, otherwise $h(\bm x(t))$ remains unchanged by (\ref{eqn:1tlce}) and Thm. \ref{thm:TL}. The only difference between the first-order TLCe (\ref{eqn:1tlce}) and the adaptive barrier function \cite{xiao2021adaptive} lies in the fact that adaptive term $\frac{1}{p(\xi)(t-t_0)}$ in the first-order TLCe depends on the future time $\xi\in(t_0,t)$, while the adaptive term in the adaptive barrier function depends on the current time. 

\textbf{High-order case.} In the high-order case (i.e., $m>1$ in (\ref{eqn:tlce})), we also assume $L_f^mh(\bm x(t_0))\ne 0$. Then the TLCe (\ref{eqn:tlce}) can be rewritten as:
\begin{equation} \label{eqn:htlce}\begin{aligned}
     \sum_{k=0}^{m-1} \frac{L_f^kh(\bm x(t_0))}{k!} (t-t_0)^k + \frac{L_f^mh(\bm x(t_0))}{m!} p(\xi)(t-t_0)^m \geq 0, \\ \xi\in(t_0, t),
\end{aligned}
\end{equation}
where
\begin{equation}
    p(\xi) = \frac{L_f^mh(\bm x(\xi))}{L_f^mh(\bm x(t_0))}.
\end{equation}

As we assume $p(\xi) > 0$, equation (\ref{eqn:htlce}) is equivalent to the high-order case of adaptive barrier function/certificate with linear class $\mathcal{K}$ functions:
\begin{equation} \label{eqn:htlce2}\begin{aligned}
     \frac{m!}{p(\xi)(t-t_0)^m}\sum_{k=0}^{m-1} \frac{L_f^kh(\bm x(t_0))}{k!} (t-t_0)^k + L_f^mh(\bm x(t_0)) \geq 0, \\ \xi\in(t_0, t),
\end{aligned}
\end{equation}
where $\frac{m!}{p(\xi)(t-t_0)^m}$ is the corresponding adaptive term. The analysis is similar to that of the first-order case.

However, the high-order TLCe is different from the high-order barrier certificate/function as the high-order TLCe may be the imaginary value high-order barrier certificate/function. Specifically, let's consider the second order case (i.e., i.e., $m=2$ in (\ref{eqn:tlce})), the corresponding TLCe (\ref{eqn:htlce2}) becomes:
\begin{equation} \label{eqn:htlce3}\begin{aligned}
     \frac{2}{p(\xi)(t-t_0)^2}(h(\bm x(t_0)) + L_fh(\bm x(t_0))(t-t_0)) \\+ L_f^2h(\bm x(t_0))\geq 0,  \xi\in(t_0, t),
\end{aligned}
\end{equation}

The high-order barrier function with linear class $\mathcal{K}$ is given by
\begin{equation} \label{eqn:hocbf2}
    L_f^2h(\bm x(t_0)) + (p_1 + p_2)L_fh(\bm x(t_0)) + p_1p_2h(\bm x(t_0)) \geq 0,
\end{equation}
where $t_0\in[0,\infty)$ and $p_1 > 0, p_2 > 0$. Comparing (\ref{eqn:hocbf2}) and (\ref{eqn:htlce3}), we have that 
\begin{equation}
    p_1 + p_2 = \frac{2}{p(\xi)(t-t_0)}, \quad p_1p_2 = \frac{2}{p(\xi)(t-t_0)^2},
\end{equation}
Solving the last two equations, we have that
\begin{equation}
    p_1 = \frac{1-\sqrt{1-2p(\xi)}}{p(\xi)(t-t_0)}, \quad p_2 = \frac{1+\sqrt{1-2p(\xi)}}{p(\xi)(t-t_0)},
\end{equation}
In the above, if $p(\xi) > \frac{1}{2}$, then $p_1, p_2$ are two imaginary numbers. This means the TLCe can extend the HOBF to the imaginary value domain, while HOBF is restricted to the real value domain (i.e., class $K$ functions are real-valued functions). Similar analysis can be applied to the higher-order case of TLCe.
\end{proof}

In fact, the TLCe is more general than existing BF/HOBF \cite{xiao2021high} as we need to have assumptions that 
$L_f^mh(\bm x(t_0))\ne 0$ and $p(\xi) > 0$ in order to reformulate a TLCe into a BF/HOBF. As TLCe extends HOBF into the imaginary value domain, the proposed TLCe pushes the boundary of barrier certificate. 

The existing of a BF/HOBF is only a sufficient condition for the safety of system (\ref{eqn:sys}). The reason that a BF/HOBF is not necessary for the safety of system (\ref{eqn:sys}) is mainly limited by the class $\mathcal{K}$ function. However, an adaptive BF/HOBF is indeed necessary and sufficient for the safety of system (\ref{eqn:sys}) since it relaxes the class $\mathcal{K}$ function by an adaptive penalty term \cite{xiao2021adaptive}, as also shown by Prop. \ref{prop:tlce}. The existence of a TLCe is a necessary and sufficient condition for the safety of system (\ref{eqn:sys}), as shown in the following theorem.

\begin{theorem}\label{thm:nesu}
    The existence of a TLCe as in Def. \ref{def:tlce} is a necessary and sufficient condition for the forward invariance of the set $C$ (i.e., the safety) for system (\ref{eqn:sys}).
\end{theorem}

\begin{proof}
     By Thm. \ref{thm:tlce}, we have that 
    \begin{equation}
\begin{aligned}
    h(t) = \sum_{k=0}^{m-1} \frac{h^{(k)}(t_0)}{k!} (t-t_0)^k + \frac{h^{(m)}(\xi)}{m!} (t-t_0)^m \\= \sum_{k=0}^{m-1} \frac{L_f^kh(\bm x(t_0))}{k!} (t-t_0)^k + \frac{L_f^mh(\bm x(\xi))}{m!} (t-t_0)^m, \\t_0\in[0,\infty), \xi\in(t_0, t),
\end{aligned}
\end{equation}

Therefore, we have that $h(\bm x(t)) = h(t)\geq 0 \Leftrightarrow$ the existence of a TLCe, for all $t$, and existence of a TLCe is a necessary and sufficient condition for the forward invariance of the set $C$ for system (\ref{eqn:sys}).
\end{proof}

\subsection{Taylor-Lagrange Control}

In this section, we consider the safety for affine control system (\ref{eqn:affine}).

We define a safe set $C$ as in (\ref{eqn:C-1}) for a safety constraint $h(\bm x)\geq 0$ that has relative degree $m$ w.r.t system (\ref{eqn:affine}). Then we propose the following definition for a Taylor-Lagrange control:
\begin{definition} [Taylor-Lagrange Control]\label{def:tlc}
    The $r\in\mathbb{N}$ times differentiable function $h:\mathbb{R}^n\rightarrow \mathbb{R}$ is a Taylor-Lagrange Control (TLC) function of relative degree $m$ for system (\ref{eqn:affine}) if
    \begin{equation} \label{eqn:tlc}
       \begin{aligned}
     \sup_{\bm u(\xi)\in U}\left[\sum_{k=0}^{m-1} \frac{L_f^kh(\bm x(t_0))}{k!} (t-t_0)^k + \frac{L_f^mh(\bm x(\xi))}{m!} (t-t_0)^m\right. \\\left.\!+ \frac{L_gL_f^{m-1}h(\bm x(\xi))\bm u(\xi)}{m!} (t\!-\!t_0)^m\right] \!\geq\! 0, t_0\in[0,\infty), \xi\in(t_0, t),
\end{aligned}
    \end{equation} for all $\bm x(t_0)\in C$,
    where $L_gh$ denotes the Lie derivative of $h$ along $g$.
\end{definition}
The interpretation of a TLC is similar to that of a TLCe, the difference is that we expand the safety constraint eventually to make the control input $\bm u(\xi),\xi\in(t_0, t)$ show up for the control system (\ref{eqn:affine}) in a TLC function.

Given a TLC as in Def. \ref{def:tlc}, we can define a state feedback controller that satisfies (\ref{eqn:tlc}):
\begin{equation}
\begin{aligned}
    K_{tlc}(\bm x(t_0), \bm x(\xi)) = \{\bm u(\xi)\in U:  \sum_{k=0}^{m-1} \frac{L_f^kh(\bm x(t_0))}{k!} (t-t_0)^k \\+ \frac{L_f^mh(\bm x(\xi))}{m!} (t-t_0)^m + \frac{L_gL_f^{m-1}h(\bm x(\xi))\bm u(\xi)}{m!} (t\!-\!t_0)^m \\\geq 0, t_0\in[0,\infty), \xi\in(t_0, t)
    \}.
\end{aligned}
\end{equation}

\begin{theorem}
\label{thm:tlc}
    Given a TLC function $h(\bm x)$ from Def. \ref{def:tlc} with the associated safe set defined as in (\ref{eqn:C-1}), if $h(\bm x(t_0))\geq 0$, then any Lipschitz continuous controller $\bm u(\xi)\in K_{tlc}(\bm x(t_0), \bm x(\xi)), \forall t_0\in[0,\infty), \xi\in(t_0,t), t>t_0$ renders the set C forward invariant for system (\ref{eqn:affine}).
\end{theorem}
\begin{proof}
Given a TLC as in Def. \ref{def:tlc}, any controller that satisfies the TLC constraint (\ref{eqn:tlc}) is equivalent to the satisfaction of the following:
\begin{equation}
       \begin{aligned}
    \sum_{k=0}^{m-1} \frac{h^{(k)}(\bm x(t_0))}{k!} (t-t_0)^k + \frac{h^{(m)}(\bm x(\xi))}{m!} (t-t_0)^m \geq 0, \\t_0\in[0,\infty), \xi\in(t_0, t),
\end{aligned}
    \end{equation}

By Thm. \ref{thm:tlce}, we have that
 \begin{equation}
\begin{aligned}
    h(t) = \sum_{k=0}^{m-1} \frac{h^{(k)}(t_0)}{k!} (t-t_0)^k + \frac{h^{(m)}(\xi)}{m!} (t-t_0)^m.
\end{aligned}
\end{equation}

Combining the last two equations, we have that $h(\bm x(t)) = h(t)\geq 0, \forall t$. Therefore, the set $C$ is forward invariant for system (\ref{eqn:affine}).
\end{proof}

The proposed TLC also has some connections with existing CBFs \cite{Ames2017} and HOCBFs \cite{xiao2021high}, and it further pushes their boundaries. The following proposition shows their relationship:

\begin{proposition} \label{lem:tlc-cbf}
    Assume $L_f^mh(\bm x(t_0)) + L_gL_f^{m-1}h(\bm x(t_0)) \bm u(t_0) \ne 0$ and $sgn(L_f^mh(\bm x(\xi)) + L_gL_f^{m-1}h(\bm x(\xi)) \bm u(\xi)) = sgn(L_f^mh(\bm x(t_0)) + L_gL_f^{m-1}h(\bm x(t_0)) \bm u(t_0))$ for a TLC in Def. \ref{def:tlc}. The first-order TLC  degenerates to an adaptive CBF with a linear class $\mathcal{K}$ function, and the high-order TLC extends the HOCBF to the imaginary value domain.
\end{proposition}
\begin{proof}
\textbf{First-order case.} When the TLC has an order $m = 1$, the corresponding first-order TLC (\ref{eqn:tlc}) becomes:
\begin{equation}
\label{eqn:1tlc}
    h(\bm x(t_0)) + L_fh(\bm x(\xi))(t-t_0) + L_gh(\bm x(\xi))\bm u(\xi)(t-t_0)\geq 0, 
\end{equation}
where $\xi\in(t_0, t).$

Since $L_f h(\bm x(t_0)) + L_gh(\bm x(t_0))\ne 0$, the above equation can be rewritten as
\begin{equation} \label{eqn:acbf21}
\begin{aligned}
    h(\bm x(t_0)) + L_fh(\bm x(t_0))p(\xi)(t-t_0) \\+ L_gh(\bm x(t_0))\bm u(t_0)p(\xi)(t-t_0)\geq 0,  \xi\in(t_0, t),
    \end{aligned}
\end{equation}
where
\begin{equation}
    p(\xi) = \frac{L_fh(\bm x(\xi))+ L_gh(\bm x(\xi))\bm u(\xi)}{L_fh(\bm x(t_0)) + L_gh(\bm x(t_0))\bm u(t_0)}.
\end{equation}

As we assume $p(\xi) > 0$, equation (\ref{eqn:acbf21}) is equivalent to the adaptive CBF \cite{xiao2021adaptive} with a linear class $\mathcal{K}$ function:
\begin{equation} \label{eqn:acbf22}
    \frac{1}{p(\xi)(t-t_0)}h(\bm x(t_0)) + L_fh(\bm x(t_0)) + L_gh(\bm x(t_0))\bm u(t_0)\geq 0,  
\end{equation}
where $\xi\in(t_0, t), \frac{1}{p(\xi)(t-t_0)}$ is the adaptive penalty function depending on the future time $\xi\in(t_0, t)$ when we consider $t_0$ as the current time. When varying $t_0$ between $[0,\infty)$, the above equation offers an adaptive CBF \cite{xiao2021adaptive}. The only difference between the first-order TLC (\ref{eqn:1tlc}) and the adaptive CBF lies in the fact that adaptive term $\frac{1}{p(\xi)(t-t_0)}$ in the first-order TLC depends on the future time $\xi\in(t_0,t)$, while the adaptive term in the adaptive CBF depends on the current time. 

\textbf{High-order case.} The high-order case is similar to that of Prop. \ref{prop:tlce}, and thus is omitted. The high-order TLC can also be related to a HOCBF. For example, in the second order case, the TLC corresponds to  $L_f^2h(\bm x(t_0)) + L_gL_fh(\bm x(t_0))\bm u(t_0) + (p_1 + p_2)L_fh(\bm x(t_0)) + p_1p_2h(\bm x(t_0)) \geq 0,$ where
\begin{equation}
    p_1 = \frac{1-\sqrt{1-2p(\xi)}}{p(\xi)(t-t_0)}, \quad p_2 = \frac{1+\sqrt{1-2p(\xi)}}{p(\xi)(t-t_0)},
\end{equation}
In the above, if $p(\xi) > \frac{1}{2}$, then $p_1, p_2$ are two imaginary numbers. This means the TLC can extend the HOCBF to the imaginary value domain, while HOCBF is restricted to the real value domain (i.e., class $K$ functions are real-valued functions). Similar analysis can be applied to the higher-order case of TLC.
\end{proof}

In HOCBFs, we need to explicitly define a sequence of safe sets, and the number of safe sets equals the relative degree of the safety constraint. However, we do not need to explicitly define those safe sets in the proposed TLC method.
The existing of a CBF/HOCBF is only a sufficient condition for the safety of system (\ref{eqn:affine}). While the existing of a TLC is a necessary and sufficient condition for the safety of system (\ref{eqn:affine}), as shown in the following theorem.

\begin{theorem} \label{thm:ns}
    The existence of a TLC function as in Def. \ref{def:tlc} is a necessary and sufficient condition for the forward invariance of the set $C$ (i.e., the safety) for system (\ref{eqn:affine}).
\end{theorem}
\begin{proof}
The proof is similar to that of Thm. \ref{thm:nesu}.
\end{proof}

\begin{remark}[Extension to time-varying systems and constraints]
When system (\ref{eqn:affine}) or the safety constraint $h(\bm x, t)\geq 0$ is time-varying (i.e., an explicit function of time), the proposed TLC works similarly as we only need to additionally consider the partial derivative of $h(\bm x, t)$ with respect to time in $h^{(k)}, k\in \{0,\dots, m\}$. This is helpful in temporal-spatial specifications \cite{srinivasan2018control}.
\end{remark}

\begin{remark}[Comparison with adaptive CBFs \cite{xiao2021adaptive}]
The proposed TLC degenerates to an adaptive CBF in the first-order case. However, the TLC is much more simpler in the high-order case. In a high-order adaptive CBF, we have to define auxiliary dynamics for the penalty functions as they will be differentiated multiple times. This makes adaptive CBF much more complicated than the proposed TLC. Moreover, there may be different penalty functions on different orders in adaptive CBFs, and there will also be more decision variables. Therefore, the computational complexity is also higher in an adaptive CBF than the proposed TLC.
\end{remark}

The necessity and sufficiency of the safety in a TLC/TLCe depends on the $\xi\in(t_0, t)$ that is generally hard to identify. However, this problem can be addressed in optimal control problems through zero-order hold (ZOH) implementation, as shown in the following sections.

\section{Taylor-Lagrange Control for Stability}
\label{sec:stability}

We also propose to employ the proposed Taylor-Lagrange Control to ensure the stability of a system. In the safety-critical control case, we usually assume that the safety constraint is initially satisfied, i.e., $h(\bm x(t_0)) \geq 0$. If this assumption is not satisfied, then we may define a new function $V(\bm x) = -h(x)$ to stabilize the system to the safe set. Without loss of generality, we consider a positive definite function $V(\bm x)$ to study the stability of system (\ref{eqn:affine}).

We formally define a Taylor-Lagrange Stability (TLS) function as follows:
\begin{definition}  [Taylor-Lagrange Stability]\label{def:tls}
	 A continuously differentiable function $V: \mathbb{R}^n\rightarrow \mathbb{R}$ is a Taylor-Lagrange Stability (TLS) function of relative degree $m$ for system (\ref{eqn:affine}) if there exist constants $c_1 >0, c_2>0$  such that, for $\forall \bm x\in \mathbb{R}^n$,
	\begin{equation}\label{eqn:tls}
    \begin{aligned}
    &c_1||\bm x||^2 \leq V(\bm x) \leq c_2 ||\bm x||^2,\\
	&\inf_{\bm u(\xi)\in U}\left[\sum_{k=0}^{m-1} \frac{L_f^kV(\bm x(t_0))}{k!} (t-t_0)^k + \frac{L_f^mV(\bm x(\xi))}{m!} (t-t_0)^m\right. \\&\left.\!+ \frac{L_gL_f^{m-1}V(\bm x(\xi))\bm u(\xi)}{m!} (t\!-\!t_0)^m\right] \!\leq\! 0, t_0\in[0,\infty), \xi\in(t_0, t),
    \end{aligned}
	\end{equation}
\end{definition}

Given a TLS as in Def. \ref{def:tls}, we can define a state feedback controller that satisfies (\ref{eqn:tls}):
\begin{equation}
\begin{aligned}
    K_{tls}(\bm x(t_0), \bm x(\xi)) = \{\bm u(\xi)\in U:  \sum_{k=0}^{m-1} \frac{L_f^kV(\bm x(t_0))}{k!} (t-t_0)^k \\+ \frac{L_f^mV(\bm x(\xi))}{m!} (t-t_0)^m + \frac{L_gL_f^{m-1}V(\bm x(\xi))\bm u(\xi)}{m!} (t\!-\!t_0)^m \\\leq 0, t_0\in[0,\infty), \xi\in(t_0, t)
    \}.
\end{aligned}
\end{equation}

\begin{theorem}
\label{thm:tls}
    Given a TLS function $V(\bm x)$ from Def. \ref{def:tls}, any Lipschitz continuous controller $\bm u(\xi)\in K_{tls}(\bm x(t_0), \bm x(\xi)), \forall t_0\in[0,\infty), \xi\in(t_0,t), t>t_0$ stabilizes the system (\ref{eqn:affine}) to the origin.
\end{theorem}
\begin{proof}
Given a TLS as in Def. \ref{def:tls}, any controller that satisfies the TLS constraint (\ref{eqn:tls}) is equivalent to the satisfaction of the following:
\begin{equation}
       \begin{aligned}
    \sum_{k=0}^{m-1} \frac{V^{(k)}(\bm x(t_0))}{k!} (t-t_0)^k + \frac{V^{(m)}(\bm x(\xi))}{m!} (t-t_0)^m \leq 0, \\t_0\in[0,\infty), \xi\in(t_0, t),
\end{aligned}
    \end{equation}

    Following explicit dynamics analysis \cite{ascher1997implicit}, the evolution of the state $\bm x$ of system (\ref{eqn:sys}) can be represented as an explicit function of time. Therefore, the TLS function $V(\bm x)$ can be rewritten as $V(t)$.
By Thm. \ref{thm:TL}, we have that
 \begin{equation}
\begin{aligned}
    V(\bm x(t)) &= V(t) \\&= \sum_{k=0}^{m-1} \frac{V^{(k)}(t_0)}{k!} (t-t_0)^k + \frac{V^{(m)}(\xi)}{m!} (t-t_0)^m \leq 0.
\end{aligned}
\end{equation}

Since $V(\bm x(t))$ is positive definite, we have that the system (\ref{eqn:affine}) is stabilized to the origin if the controller satisfies the TLS constraint (\ref{eqn:tls}).
\end{proof}

The proposed TLS has some connections with existing CLFs \cite{Aaron2012} and HOCLFs \cite{xiao2021adaptive}, and it further pushes their boundaries. The following proposition shows their relationship:

\begin{proposition} \label{lem:tls-clf}
    Assume $L_f^mV(\bm x(t_0)) + L_gL_f^{m-1}V(\bm x(t_0)) \bm u(t_0) \ne 0$ and $sgn(L_f^mV(\bm x(\xi)) + L_gL_f^{m-1}V(\bm x(\xi)) \bm u(\xi)) = sgn(L_f^mV(\bm x(t_0)) + L_gL_f^{m-1}V(\bm x(t_0)) \bm u(t_0))$ for a TLC in Def. \ref{def:tls}. The first-order TLC  degenerates to a CLF, and the high-order TLS extends the HOCLF to the imaginary value domain.
\end{proposition}
\begin{proof}
\textbf{First-order case.} When the TLS has an order $m = 1$, the corresponding first-order TLC (\ref{eqn:tls}) becomes:
\begin{equation}
\label{eqn:1tls}
    V(\bm x(t_0)) + L_fV(\bm x(\xi))(t-t_0) + L_gV(\bm x(\xi))\bm u(\xi)(t-t_0)\leq 0, 
\end{equation}
where $\xi\in(t_0, t).$

Since $L_f V(\bm x(t_0)) + L_gV(\bm x(t_0))\ne 0$, the above equation can be rewritten as
\begin{equation} \label{eqn:aclf21}
\begin{aligned}
    V(\bm x(t_0)) + L_fV(\bm x(t_0))p(\xi)(t-t_0) \\+ L_gV(\bm x(t_0))\bm u(t_0)p(\xi)(t-t_0)\leq 0,  \xi\in(t_0, t),
    \end{aligned}
\end{equation}
where
\begin{equation}
    p(\xi) = \frac{L_fV(\bm x(\xi))+ L_gV(\bm x(\xi))\bm u(\xi)}{L_fV(\bm x(t_0)) + L_gV(\bm x(t_0))\bm u(t_0)}.
\end{equation}

As we assume $p(\xi) > 0$, equation (\ref{eqn:aclf21}) is equivalent to the following CLF:
\begin{equation} \label{eqn:aclf22}
    \frac{1}{p(\xi)(t-t_0)}V(\bm x(t_0)) + L_fV(\bm x(t_0)) + L_gV(\bm x(t_0))\bm u(t_0)\leq 0,  
\end{equation}
where $\xi\in(t_0, t), \frac{1}{p(\xi)(t-t_0)}$ is the adaptive penalty function depending on the future time $\xi\in(t_0, t)$ when we consider $t_0$ as the current time. When varying $t_0$ between $[0,\infty)$, the above equation offers a CLF condition. 

\textbf{High-order case.} The high-order case is similar to that of Prop. \ref{prop:tlce}, and thus is omitted. The high-order TLS can also be related to a HOCLF. For example, in the second order case, the TLS corresponds to  $L_f^2V(\bm x(t_0)) + L_gL_fV(\bm x(t_0))\bm u(t_0) + (p_1 + p_2)L_fV(\bm x(t_0)) + p_1p_2V(\bm x(t_0)) \leq 0,$ where
\begin{equation}
    p_1 = \frac{1-\sqrt{1-2p(\xi)}}{p(\xi)(t-t_0)}, \quad p_2 = \frac{1+\sqrt{1-2p(\xi)}}{p(\xi)(t-t_0)},
\end{equation}
In the above, if $p(\xi) > \frac{1}{2}$, then $p_1, p_2$ are two imaginary numbers. This means the TLS can extend the HOCLF to the imaginary value domain, while HOCLF is restricted to the real value domain (i.e., class $K$ functions are real-valued functions). Similar analysis can be applied to the higher-order case of TLS.
\end{proof}

Note that in a TLS we may only require $V(\bm x(t)) = 0$ (i.e., the TLS constraint (\ref{eqn:tls}) is an equality) instead of $V(\bm x(t)) \leq 0$ as  $V(\bm x)$ is positive definite. We define it in the form $V(\bm x(t)) \leq 0$ in order to make the method robust to errors, especially due to sampling, as shown in the next section.

\section{Optimal Control with Taylor-Lagrange}
\label{sec:oc}

Consider a constrained optimal control problem for system (\ref{eqn:affine}) with the cost defined as:
\begin{equation}\label{eqn:cost}
J(\bm u(t)) = \int_{0}^{T}\mathcal{C}(||\bm u(t)||)dt
\end{equation}
where $||\cdot||$ denotes the 2-norm of a vector. $T$ denotes the final time, and  $\mathcal{C}(\cdot)$ is a strictly increasing function of its argument (such as the energy consumption function $\mathcal{C}(||\bm u(t)||) = ||\bm u(t)||^2$).

We also wish the state $\bm x$ of system (\ref{eqn:affine}) to reach a desired state $\bm x_d\in\mathbb{R}^n$ at the final time $T$. Mathematically, we wish to achieve the following:
\begin{equation} \label{eqn:convergence}
    \min_{\bm u} ||\bm x(T) - \bm x_d||.
\end{equation}

Assume a safety constraint 
\begin{equation} \label{eqn:safety}
h(\bm x) \geq 0,
\end{equation}
with relative degree $m$ that has to be satisfied by system (\ref{eqn:affine}). 

Formally, we have the following constrained optimal control problems:
\begin{problem} \label{problem}
Our objective is to find a controller $\bm u^*$ for (\ref{eqn:affine}) by solving the following optimization:
\begin{equation}
\begin{aligned}
    \bm u^* &= \arg\min_{\bm u\in U} \int_{0}^{T}\mathcal{C}(||\bm u(t)||)dt + w||\bm x(T) - \bm x_d||,\\
    &\text{ s.t. (\ref{eqn:safety}),} 
\end{aligned}
\end{equation}
where $w > 0$. We may consider multiple safety constraints in the above, but we only list one for simplicity without loss of generality.
\end{problem}

\textbf{Approach:} We use a TLC to enforce the safety constraint (\ref{eqn:safety}), employ a TLS to enforce the state convergence objective (\ref{eqn:convergence}), and directly take the cost (\ref{eqn:cost}) as the objective to reformulate Problem \ref{problem} as a constrained optimization problem that take the control $\bm u$ as the decision variable. However, the TLC (\ref{eqn:tlc}) and TLS (\ref{eqn:tls}) depend on the control at the future time $\xi\in (t_0, t)$. We propose the following methods to address this problem.

\subsection{Zero-Order Hold Taylor-Lagrange Control}
For continuous optimal control problems, we have to implement the control in discrete time. Suppose we take $t_0\in [0, T]$ as the current implementation time instant, and $t = t_1$ as the next implementation time instant. We further assume a Zero-Order Hold (ZOH) policy for the discrete time implementation, which is widely adopted in the CBF method \cite{Ames2017}. Then we have that $\bm u(t_0) = \bm u(\xi), \xi\in (t_0, t_1)$. Finally, we have the following ZOH Taylor-Lagrange Control function definition:
\begin{definition} [Zero-Order Hold Taylor-Lagrange Control]\label{def:zoh-tlc}
    The $r\in\mathbb{N}$ times differentiable function $h:\mathbb{R}^n\rightarrow \mathbb{R}$ is a Zero-Order Hold Taylor-Lagrange Control (ZOH-TLC) function of relative degree $m$ for system (\ref{eqn:affine}) if
    \begin{equation} \label{eqn:zoh-tlc}
       \begin{aligned}
     \sup_{\bm u(t_0)\in U}\left[\sum_{k=0}^{m-1} \frac{L_f^kh(\bm x(t_0))}{k!} {\Delta t}^k + \frac{L_f^mh(\bm x(t_0))}{m!} {\Delta t}^m\right. \\\left.\!+ \frac{L_gL_f^{m-1}h(\bm x(t_0))\bm u(t_0)}{m!} {\Delta t}^m\right] \!\geq\! 0,
\end{aligned}
    \end{equation} for all $\bm x(t_0)\in C$,
    where $\Delta t = t_1 - t_0$.
\end{definition}

In Def. \ref{def:zoh-tlc}, we have also assumed that $\bm x(\xi) = \bm x(t_0)$. The error due to this assumption is small if $\Delta t$ is small. Again, this assumption falls into the ZOH assumption, and we can address this issue using sampled data methods \cite{breeden2021control} or event-triggered methods \cite{xiao2022event} \cite{taylor2020safety}, which will be further explored in the following subsections.

In the first-order case, the ZOH-TLC constraint (\ref{eqn:zoh-tlc}) becomes:
\begin{equation}
\label{eqn:1zoh-tlc}
     L_fh(\bm x(t_0)) + L_gh(\bm x(t_0))\bm u(t_0) + \frac{1}{\Delta t}h(\bm x(t_0))\geq 0, 
\end{equation}
which is equivalent to a first-order CBF with a linear class $\mathcal{K}$ function. {\it The Taylor-Lagrange theorem further shows that the coefficient of this class $\mathcal{K}$ function has to be $\frac{1}{\Delta t}$ in order for the existence of a CBF being a necessary and sufficient condition for the safety of the system.} 

In the extreme case as $\Delta t\rightarrow 0$ such that the ZOH approximation error can be eliminated, we see that the coefficient of the linear class $\mathcal{K}$ function goes to infinity in order for a CBF to be sufficient and necessary achieving safety. This is also consistent with the optimal control theory \cite{Bryson1969} that shows the safety constraint will not affect the optimal control solution until it becomes active (i.e., $h(\bm x) = 0$).

The more interesting analysis is in high-order cases. In the second order case, the ZOH-TLC constraint (\ref{eqn:zoh-tlc}) becomes:
\begin{equation} \begin{aligned}
      L_f^2h(\bm x(t_0)) + L_gL_fh(\bm x(t_0))\bm u(t_0) +  \frac{2}{{\Delta t}}L_fh(\bm x(t_0)) \\+ \frac{2}{{\Delta t}^2}h(\bm x(t_0)) \geq 0.
\end{aligned}
\end{equation}
Comparing the last equation with a second-order HOCBF $L_f^2h(\bm x(t_0)) + L_gL_fh(\bm x(t_0))\bm u(t_0) +  (p_1 + p_2)L_fh(\bm x(t_0)) + p_1p_2h(\bm x(t_0)) \geq 0.$ We have that:
\begin{equation}
    p_1 = \frac{1-1i}{\Delta t}, \quad p_2 = \frac{1+1i}{\Delta t},
\end{equation}
where $i$ denotes the imaginary axis, and the corresponding $p_1, p_2$ are two imaginary numbers in a HOCBF. This shows that the TLC extends the HOCBF to the imaginary value domain, and the corresponding safe set $C_2:=\{\bm x\in \mathbb{R}^n: L_h(\bm x) + p_1L_h(\bm x) \geq 0\}$ in a HOCBF is also defined in the complex domain. On the other hand, the TLC shows that a necessary and sufficient HOCBF for the safety of system (\ref{eqn:affine}) is defined in the complex domain.

\begin{remark} [Higher-Order ZOH-TLC]
    Given a safety constraint $h(\bm x)\geq 0$ that has relative degree $m$, we may also define a ZOH-TLC with relative degree $m+1$. Then the control $\bm u(t_0)$ would show up in $h^{(m)}(\bm x(t_0))$ in (\ref{eqn:tlc}), while $\bm u(\xi), \dot {\bm u}(\xi), \xi\in(t_0, t)$ would show up in $h^{(m+1)}(\bm x(\xi))$. Since we assume ZOH, we have that $\dot {\bm u}(\xi) = 0$. Since system (\ref{eqn:affine}) is Lipschitz and $\bm u$ is bounded, we can find the bound $M$ for $h^{(m+1)}(\bm x(\xi))$. This allows us to define robust TLC. We will further study this in future work.
\end{remark}

\subsection{Zero-Order Hold Taylor-Lagrange for Stability}

As in the ZOH-TLC case, we take $t_0\in [0, T]$ as the current implementation time instant, and $t = t_1$ as the next implementation time instant, and assume the discretization time interval $\Delta = t_1-t_0$ to be small enough. Then we have the ZOH assumption that $\bm u(t_0) = \bm u(\xi), \bm x(t_0) = \bm x(\xi), \xi\in (t_0, t_1)$ for a TLS as in Def. \ref{def:tls}.
Finally, we have the following ZOH Taylor-Lagrange Stability function definition:
\begin{definition}  [Zero-Order Hold Taylor-Lagrange Stability]\label{def:zoh-tls}
	 A continuously differentiable function $V: \mathbb{R}^n\rightarrow \mathbb{R}$ is a Zero-Order Hold Taylor-Lagrange Stability (ZOH-TLS) function of relative degree $m$ for system (\ref{eqn:affine}) if there exist constants $c_1 >0, c_2>0$  such that, for $\forall \bm x\in \mathbb{R}^n$,
	\begin{equation}\label{eqn:zoh-tls}
    \begin{aligned}
    &c_1||\bm x||^2 \leq V(\bm x) \leq c_2 ||\bm x||^2,\\
	&\inf_{\bm u(t_0)\in U}\left[\sum_{k=0}^{m-1} \frac{L_f^kV(\bm x(t_0))}{k!} {\Delta t}^k + \frac{L_f^mV(\bm x(t_0))}{m!} \Delta t^m\right. \\&\left.+ \frac{L_gL_f^{m-1}V(\bm x(t_0))\bm u(t_0)}{m!} \Delta t^m\right] \leq 0. 
    \end{aligned}
	\end{equation}
\end{definition}

In the first-order case, the ZOH-TLC constraint (\ref{eqn:zoh-tls}) becomes:
\begin{equation}
\label{eqn:1zoh-tls}
     L_fV(\bm x(t_0)) + L_gV(\bm x(t_0))\bm u(t_0) + \frac{1}{\Delta t}V(\bm x(t_0))\leq 0, 
\end{equation}
which is equivalent to a CLF in Def. \ref{def:clf} with an explicit coefficient $c_3 = \frac{1}{\Delta t}$ on the $V(\bm x)$ function. If the discretization time interval $\Delta t$ becomes smaller, then the system state will converge to the origin faster in a ZOH-TLS. The analysis of high-order cases is similar to that of ZOH-TLC.

\subsection{ ZOH-TLC-TLS based Quadratic Programs}

Eventually, we consider solving the Problem \ref{problem} using the proposed ZOH-TLC and ZOH-TLS. We use a ZOH-TLC to enforce the safety constraint (\ref{eqn:safety}), and use a ZOH-TLS $V(\bm x) = (\bm x - \bm x_d)^2$ to enforce the state convergence requirement (\ref{eqn:convergence}). The $\mathcal{C}(\cdot)$ function in (\ref{eqn:cost}) usually takes a quadratic form, and we can formluate the following ZOH-TLC-TLS based Quadratic Programs at the current time $t_0\in[0,T]$:
\begin{equation}\label{eqn:qp}
    \min_{\bm u(t_0),\delta(t_0)} ||\bm u(t_0)||^2 + w\delta(t_0)^2
\end{equation}
s.t.
$$
\begin{aligned}
\sum_{k=0}^{m-1} \frac{L_f^kh(\bm x(t_0))}{k!} {\Delta t}^k + \frac{L_f^mh(\bm x(t_0))}{m!} {\Delta t}^m \\+ \frac{L_gL_f^{m-1}h(\bm x(t_0))\bm u(t_0)}{m!} {\Delta t}^m\!\geq\! 0,
\end{aligned}
$$
$$
    \begin{aligned}
    \sum_{k=0}^{m-1} \frac{L_f^kV(\bm x(t_0))}{k!} {\Delta t}^k + \frac{L_f^mV(\bm x(t_0))}{m!} \Delta t^m \\+ \frac{L_gL_f^{m-1}V(\bm x(t_0))\bm u(t_0)}{m!} \Delta t^m \leq \delta(t_0), 
    \end{aligned}
$$
$$\bm u_{min}\leq \bm u(t_0)\leq \bm u_{max}.$$
where $\delta(t_0)$ is a slack variable to address potential conflict between the ZOH-TLC and ZOH-TLS constraints. At the current time $t_0\in [0,T]$, we take $\bm x(t)$ as a constant within the interval $[t_0, t_1]$, and the above problem becomes a Quadratic Program (QP). We solve the above QP and get an optimal control $\bm u^*(t_0)$ that is used to update the dynamics (\ref{eqn:affine}) for $t\in[t_0,t_1]$, and then set $t_0\leftarrow t_1$. We repeat the procedure until the final time $T$. 

The inter-sampling effect is important in the proposed TLC method as the corresponding TLC constraint (\ref{eqn:tlc}) is equivalent to the original safety constraint $h(\bm x)\geq 0$ by Thm. \ref{thm:ns}. This means that the TLC will only work when the system state is at the boundary of the unsafe set at which the inter-sampling effect is serious. In order to address this, we propose an event-triggered TLC method as shown in the next subsection.

\subsection{Event-triggered ZOH-TLC based QPs}

The event-triggered method for addressing the inter-sampling effect has been extensively studied in the literature \cite{xiao2022event} \cite{taylor2020safety}. We show here how to address the inter-sampling effect of the proposed TLC method.

We first define a state bound $S(\bm x(t_0))$ at the current event time $t_0$ in the form:
\begin{equation} \label{eqn:state_bnd}
    S(\bm x(t_0)) := \{\bm y\in \mathbb{R}^n: \bm x(t_0) - \bm x_{lower} \leq \bm y\leq \bm x(t_0) + \bm x_{up}\},
\end{equation}
where $\bm x_{lower}\in\mathbb{R}^n, \bm x_{up}\in\mathbb{R}^n$, and the above inequality is interpreted component-wise.

In order to find a robust TLC for $\bm x\in S(\bm x(t_0))$, we define the following:
\begin{equation}\label{eqn:r-b}
    h_{rtlc} = \min_{\bm x\in S(\bm x(t_0))} \sum_{k=0}^{m-1} \frac{L_f^kh(\bm x)}{k!} {\Delta t}^k + \frac{L_f^mh(\bm x)}{m!} {\Delta t}^m
\end{equation}

\begin{equation} \label{eqn:r-ctrl}
\begin{aligned}
G_{rtlc} &= (G_{trlc,1}, G_{trlc,2}, \dots, G_{trlc,q}), \\
G_{rtlc,k} &= \left\{\begin{array}{c} 
\min_{\bm x\in S(\bm x(t_0))}: \frac{[L_gL_f^{m-1}h(\bm x)]_k {\Delta t}^m}{m!}, \text{ if } u_k \geq 0, \\
\max_{\bm x\in S(\bm x(t_0))}: \frac{[L_gL_f^{m-1}h(\bm x)]_k {\Delta t}^m}{m!}, \text{ otherwise},
\end{array} \right.,
\end{aligned}
\end{equation}
where $k\in \{1,\dots, q\},$ $\bm u = (u_1, \dots, u_q)$, and $[L_gL_f^{m-1}h(\bm x)]_k$ denotes the $k_{th}$ component of the vector $L_gL_f^{m-1}h(\bm x)$.  The sign of $u_k$ can be determined by solving the QP (\ref{eqn:qp}) at the current time $t_0$.

We define the event-triggered TLC as follows:
\begin{definition} [Event-Triggered Taylor-Lagrange Control]\label{def:event-tlc}
    Let $S_{\bm x(t_0)}$ be defined by (\ref{eqn:state_bnd}) and $G_{rtlc}$ and $h_{rtlc}$ be defined by (\ref{eqn:r-b}), (\ref{eqn:r-ctrl}). The $r\in\mathbb{N}$ times differentiable function $h:\mathbb{R}^n\rightarrow \mathbb{R}$ is an event-triggered Taylor-Lagrange control  function of relative degree $m$ for system (\ref{eqn:affine}) if
    \begin{equation} \label{eqn:event-tlc}
       \begin{aligned}
     \sup_{\bm u(t_0)\in U}\left[G_{rtlc}\bm u(t_0) + h_{rtlc}\right] \geq 0,
\end{aligned}
    \end{equation} for all $\bm x(t_0)\in C$.
\end{definition}

Now we have the following new event-triggered QP:
\begin{equation} \label{eqn:event-qp}
    \min_{\bm u(t_0),\delta(t_0)} ||\bm u(t_0)||^2 + w\delta(t_0)^2
\end{equation}
s.t.
$$
\begin{aligned}
G_{rtlc}\bm u(t_0) + h_{rtlc}\geq 0,
\end{aligned}
$$
$$
    \begin{aligned}
    \sum_{k=0}^{m-1} \frac{L_f^kV(\bm x(t_0))}{k!} {\Delta t}^k + \frac{L_f^mV(\bm x(t_0))}{m!} \Delta t^m \\+ \frac{L_gL_f^{m-1}V(\bm x(t_0))\bm u(t_0)}{m!} \Delta t^m \leq \delta(t_0), 
    \end{aligned}
$$
$$\bm u_{min}\leq \bm u(t_0)\leq \bm u_{max}.$$
Note that we can define similar event-triggered TLS in the above QP. 

The next time $t_1$ to solve the QP is then determined by:
\begin{equation}\label{eqn:event_time}
    t_1 = \{t>t_0: \bm x(t) \notin S(\bm x(t_0))\},
\end{equation}
Note that $t_1 - t_0$ does not equal to the $\Delta t$ in (\ref{eqn:event-tlc}) as the $\Delta t$ is used to define the TLC while $t_1 - t_0$ is the event time interval. When the next event time $t_1$ is triggered, we take $t_0 \leftarrow t_1$, and resolve the QP (\ref{eqn:event-qp}).

We have the following theorem to show the safety of the event-triggered TLC based QP:
\begin{theorem}
    Given an event-triggered TLC defined as in Def. \ref{def:event-tlc}, if the next event time $t_1$ of solving the QP (\ref{eqn:event-qp}) is determined by (\ref{eqn:event_time}), then the set $C$ is forward invariant for system (\ref{eqn:affine}).
\end{theorem}
\begin{proof}
By (\ref{eqn:r-b}) and (\ref{eqn:r-ctrl}), we have that
$$
\begin{aligned}
\sum_{k=0}^{m-1} \frac{L_f^kh(\bm x(t))}{k!} {\Delta t}^k + \frac{L_f^mh(\bm x(t))}{m!} {\Delta t}^m \\+ \frac{L_gL_f^{m-1}h(\bm x(t))\bm u(t_0)}{m!} {\Delta t}^m\geq \\G_{rtlc}\bm u(t_0) + h_{rtlc}, \forall t\in [t_0, t_1]
\end{aligned}
$$
where $t_1$ is the next event time defined by (\ref{eqn:event_time}) and $\bm x(t)\in S(\bm x(t_0)), \forall t\in [t_0, t_1]$.

As the control $\bm u(t_0)$ is obtained by solving the QP (\ref{eqn:event-qp}), we have that 
$$
G_{rtlc}\bm u(t_0) + h_{rtlc} \geq 0,
$$

Combining the last two equations, we have that 
$$
\begin{aligned}
\sum_{k=0}^{m-1} \frac{L_f^kh(\bm x(t))}{k!} {\Delta t}^k + \frac{L_f^mh(\bm x(t))}{m!} {\Delta t}^m \\+ \frac{L_gL_f^{m-1}h(\bm x(t))\bm u(t)}{m!} {\Delta t}^m\geq 0,
\end{aligned}
$$
$\forall t\in [t_0, t_1]$. By Thm. \ref{thm:tlc}, we have that $h(\bm x(t))\geq 0, \forall t\in[t_0, t_1]$ when we take every $t$ as the initial time in Thm. \ref{thm:tlc}. Thus, the set $C$ is forward invariant for system (\ref{eqn:affine}).

\end{proof}

\section{CASE STUDIES AND RESULTS}
\label{sec:case}
In this section, we present case studies for the Adaptive Cruise Control (ACC) and a 2D robot control problem. All the computations were conducted in MATLAB, and we used quadprog to solve QPs and used ode45 to integrate the dynamics.

\subsection{ACC Problem}

Consider the adaptive cruise control (ACC) problem with the vehicle dynamics in the form:
\begin{equation} \label{eqn:simpledynamics}
\left[\begin{array}{c} 
\dot v(t)\\
\dot z(t)
\end{array} \right]=
\left[\begin{array}{c}  
\frac{-F_r(v(t))}{M}\\
v_0 - v(t)
\end{array} \right] + 
\left[\begin{array}{c}  
\frac{1}{M}\\
0
\end{array} \right]u(t),
\end{equation}
where $v(t)$ denotes the velocity of the ego vehicle along its lane, $z(t)$ denotes the along-lane distance between the ego and its preceding vehicles, $v_0 > 0$ denotes the speed of the preceding vehicle, $M$ is the mass of the ego vehicle, and $u(t)$ is the control input of the controlled vehicle. $F_{r}(v(t))$ denotes the resistance force, which is expressed \cite{Khalil2002} as:
\begin{equation}\label{eqn:resistence}
F_{r}(v(t)) = f_0sgn(v(t)) + f_1v(t) + f_2 v^2(t),
\end{equation}
where $f_0 > 0, f_1 > 0$ and $f_2 > 0$ are scalars determined empirically. The first term in $F_{r}(v(t))$ denotes the coulomb friction force, the second term denotes the viscous friction force and the last term denotes the aerodynamic drag.

\textbf{Objective 1.} We wish to minimize the acceleration of the ego vehicle in the form:
\begin{equation}
    \min_u \int_0^T \left(\frac{u-F_r(v)}{M}\right)^2 dt.
\end{equation}

\textbf{Objective 2.} The ego vehicle seeks to achieve a desired speed $v_d > 0$.

\textbf{Safety.} We require that the distance $z(t)$ between the controlled vehicle and its immediately preceding vehicle be greater than a constant $\delta > 0$ for all the times, i.e.,
\begin{equation} \label{eqn:safety_acc}
z(t) \geq c,\quad\forall t\geq 0,
\end{equation}
where $c > 0$.

\textbf{Control limitation.} We also consider a control constraint $-c_dMg\leq u(t)\leq c_aMg, g = 9.81m/s^2, c_a > 0, c_d > 0$ for (\ref{eqn:simpledynamics}). 

The ACC problem is to find a control policy that achieves Objectives 1 and 2 subject to the safety constraint and control bound.
The relative degree of (\ref{eqn:safety_acc}) is two, and we use either a second order HOCBF, TLC, and event-triggered TLC to implement it by defining $h(\bm x) = z - c \geq 0$ and:
\begin{equation} \label{eqn:hocbf-case}
\begin{aligned}
    L_f^2h(\bm x) + L_gL_fh(\bm x)u + (p_1+p_2)L_fh(\bm x) \\+ p_1p_2h(x)\geq 0. 
\end{aligned}
\quad\text{HOCBF}
\end{equation}

\begin{equation} \label{eqn:tlc-case}\begin{aligned}
      L_f^2h(\bm x) + L_gL_fh(\bm x)u +  \frac{2}{{\Delta t}}L_fh(\bm x) \\+ \frac{2}{{\Delta t}^2}h(\bm x) \geq 0.
\end{aligned}
\qquad\text{TLC}
\end{equation}
The state bound for the event-triggered TLC is defined as 
$$S(\bm x) = \{\bm y\in \mathbb{R}^2: \bm x - \bm x_{lower}\leq \bm y\leq  \bm x + \bm x_{up}\}.$$
We employ a TLS (similar to a CLF) with relative degree one to enforce the desired speed. All the simulation parameters are shown in Table \ref{table:param}, in which $d_t$ denotes the sensor monitoring time gap for the event-trigger method.

\begin{table}
 	\caption{Simulation parameters for ACC}\label{table:param}
 	\begin{center}
 		\begin{tabular}{|c||c||c|c||c||c|}
 			\hline
Parameter & Value & Units &Parameter & Value & Units\\
\hline
\hline
$v(0)$ & 24& $m/s$&	$z(0)$ & 90& $m$\\
\hline
$v_{0}$ & 13.89& $m/s$ & $v_d$ & 24& $m/s$\\
\hline
$M$ & 1650& $kg$ &g & 9.81& $m/s^2$\\
\hline
$f_0$ & 0.1& $N$ &$f_1$ & 5& $Ns/m$\\
\hline
$f_2$ & 0.25& $Ns^2/m$ &$c$ & 10& $m$\\
\hline
$\Delta t$ & 0.1& $s$&	$d_t$ & 0.03& $s$\\
\hline
$c_a$ & 0.4& unitless&$c_d$ & 0.7& unitless\\
\hline
$\bm x_{lower}$ & (0.5, 1)& unitless&$\bm x_{up}$ & (0.5,1)& unitless\\
\hline
 		\end{tabular}
 	\end{center}
 	
 \end{table}

 We first present a comparison between the time-driven HOCBF and TLC when $\Delta t = 1s$, as shown in Fig. \ref{fig:time-compare}. The TLC method generates a larger deceleration than the HOCBF when the safety constraint is about to be active, and both are subject to the inter-sampling effect as the $\Delta t$ is too large, and thus the safety constraint is slightly violated.

 \begin{figure}[thpb]
	\vspace{-1mm}
	\centering
	\includegraphics[scale=0.5]{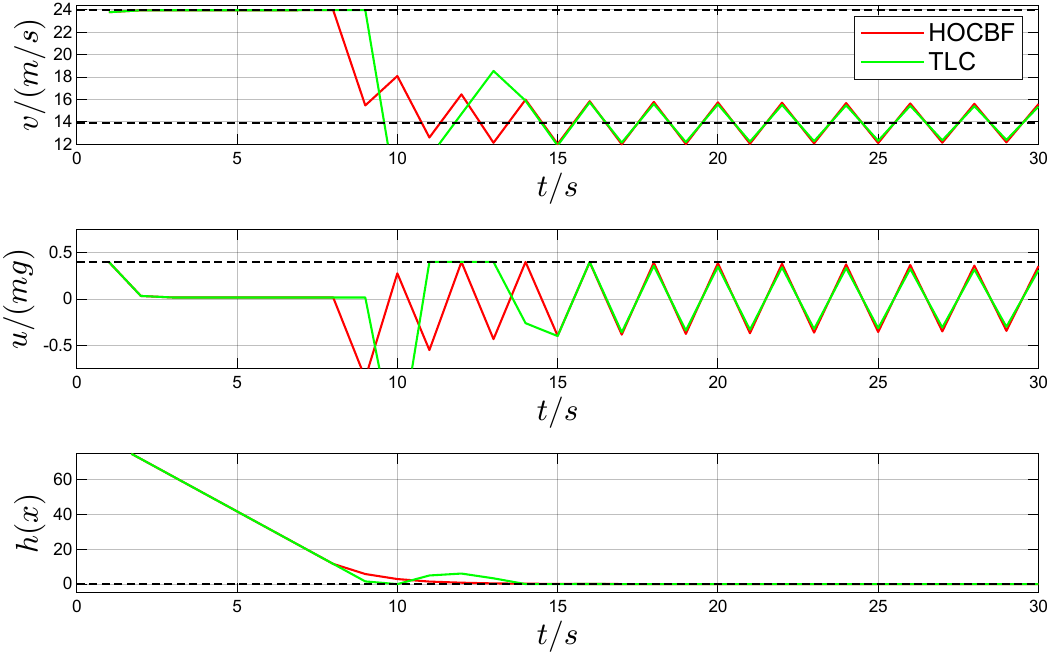}
	\vspace{-4mm}
	\caption{Comparison of State, control and safety function $h(\bm x)$ between time-driven HOCBF and TLC when $\Delta t = 1s$.}
	\label{fig:time-compare}	
	\vspace{-3mm}
\end{figure}

Then we present the comparison when the implementation time becomes $0.1s$ as shown in Fig. \ref{fig:event-compare}. The time-driven TLC still slightly violate the safety constraint due to the inter-sampling issue, while the event-triggered TLC can strictly enforce the safety. The event is more frequently triggered when the safety constraint is close to be active, as shown in the last frame of Fig. \ref{fig:event-compare}. We further show in Fig. \ref{fig:sets} how the corresponding $\psi_1(\bm x) = L_fh(\bm x) + \frac{1-1i}{\Delta t}h(\bm x)$ varies in the complex plane for the TLC and event-triggered TLC when compared with the HOCBF. The HOCBF only evolves on the real axis, while both the TLC and event-triggered TLC moves in the real-imaginary complex plane. This shows that the TLC extends the HOCBF to the imaginary value domain. Moreover, the TLC is less restrictive than the HOCBF as it only requires the forward invariance of $C_1$ while the HOCBF requires the forward invariance of $C_1\cap C_2$.

\begin{figure}[thpb]
	\vspace{-1mm}
	\centering
	\includegraphics[scale=0.5]{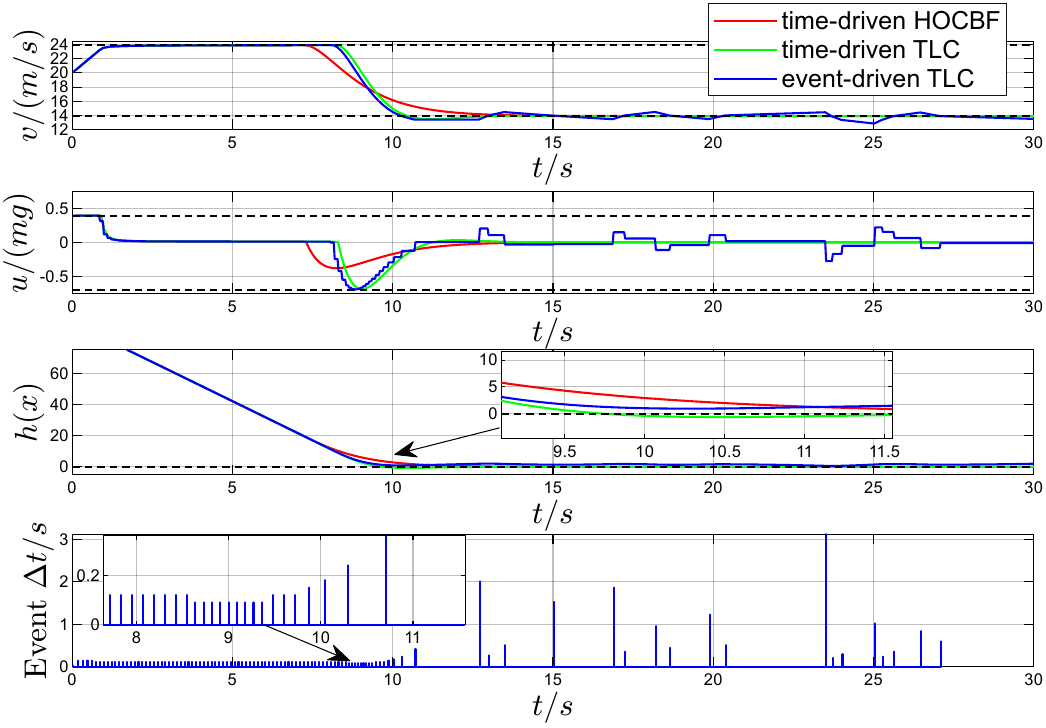}
	\vspace{-4mm}
	\caption{Comparison of State, control, safety function $h(\bm x)$, and event-triggering time $\Delta t$ between time-driven HOCBF and time-driven TLC, and event-triggered TLC when the implementation time is $0.1s$.}
	\label{fig:event-compare}	
\end{figure}

\begin{figure}[thpb]
	\centering
	\includegraphics[scale=0.7]{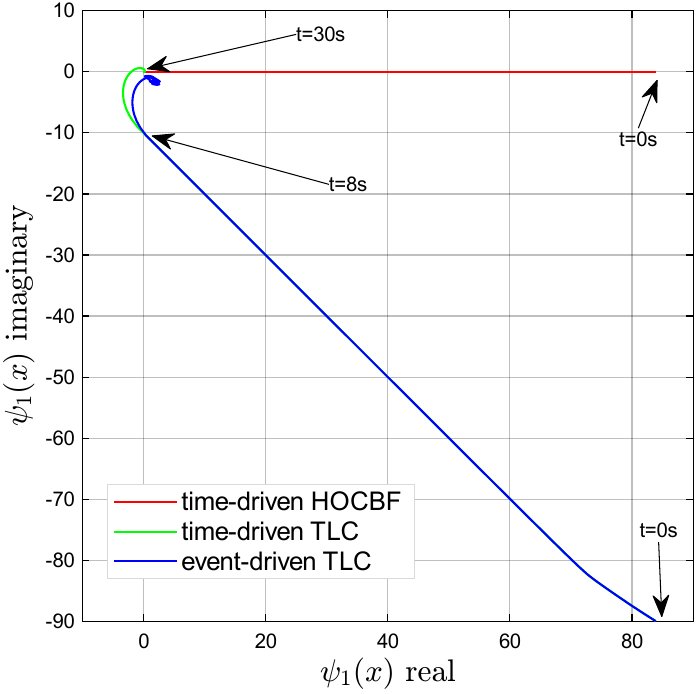}
	\caption{The variations of the corresponding $\psi_1(\bm x) = L_fh(\bm x) + \frac{1-1i}{\Delta t}h(\bm x)$ for TLC and event-triggered TLC when compared with the HOCBF in the complex plane.}	
	\label{fig:sets}
\end{figure}

\subsection{Robot Control Problem}

\textbf{Robot Dynamics:} We consider the following unicycle model for a wheeled mobile robot:
\begin{equation}\label{eqn:robot}
\begin{aligned}
\dot x = v\cos\theta,\quad
\dot y = v\sin\theta,\quad
\dot v = u_2, \quad\dot\theta = u_1,
\end{aligned}
\end{equation}
where $\bm x := (x,y,\theta,v), \bm u = (u_1,u_2)$, $x, y$ denote the location along $x, y$ axis, respectively, $\theta$ denotes the heading angle of the robot, $v$ denotes the linear speed, and $u_1, u_2$ denote the two control inputs for turning and acceleration, respectively. 

\textbf{Objective 1} (Minimum Energy Consumption) considers the cost function in the form
$$
J(\bm u(t)) = \int_{t_0}^{t_f} (u_1^2(t) + u_2^2(t)) dt.
$$

\textbf{Objective 2} (Destination) wishes the robot to arrive at a destination $(x_d, y_d)\in\mathbb{R}^2$, during time interval $[t_1,t_2], 0\leq t_1\leq t_2\leq T$.

\textbf{Constraint 1} (Safety constraint) requires the robot to satisfy the safety constraint $
(x(t) - x_o)^2 + (y(t) - y_o)^2 \geq r^2,
$
where $(x_o, y_o)\in\mathbb{R}^2$ denotes the location of the obstacle, $r = 7m$ is determined by the size of the obstacle (usually a little bigger than the size of the obstacle $r = 6m$).

\textbf{Constraint 2} (Robot Limitations) are the speed and control limitations:
$
v_{min} \leq v(t) \leq v_{max},
u_{1,min} \leq u_1(t) \leq u_{1,max},u_{2,min} \leq u_2(t) \leq u_{2,max},
$
where $v_{min}= 0m/s, v_{max}= 2m/s$, $u_{1,max}= -u_{1,min} = 0.4rad/s$, and $u_{2,max} = -u_{2,min} = 0.8m/s^2$.

We wish to determine control laws to achieve Objectives 1, 2 subject to Constraints 1, 2 (both are strictly satisfied), for the robot governed by dynamics (\ref{eqn:robot}).

We employ either HOCBFs, TLCs or event-triggered TLCs to impose Constraints 1 and 2 on control input and a TLS to achieve Objective 2. We capture Objective 1 in the cost of the optimization problem. 
Specifically, for the Constraint 1, we define $h(\bm x) = (x - x_o)^2 + (y - y_o)^2 - r^2$. The corresponding HOCBF and TLC constraints are the same as in (\ref{eqn:hocbf-case}) and (\ref{eqn:tlc-case}). For the speed limit in constraint 2, we employ a TLC with relative degree one to enforce it.

For Objective 2, we employ either CLFs or TLS to enforce it. Specifically, in the CLF method, we define two CLFs $V_1(\bm x) = (v - v_d)^2$ and $V_2(\bm x) = (\theta - arctan(\frac{y-y_d}{x-x_d}))$, where $v_d$ is a desired speed dependent on the distance to the target. In the TLS method, we define a TLS $V(\bm x) = (x-x_d)^2 + (y - y_d)^2$, and it has relative degree 2 with repect to the dynamics (\ref{eqn:robot}). The corresponding TLS constraint is then given by
\begin{equation} \begin{aligned}
      L_f^2V(\bm x) + L_gL_fV(\bm x)u +  \frac{2}{{\Delta t}}L_fV(\bm x) \\+ \frac{2}{{\Delta t}^2}V(\bm x) \leq 0.
\end{aligned}
\qquad\text{TLS}
\end{equation}

The state bound in the event-triggered TLC is defined as:
$$S(\bm x) = \{\bm y\in \mathbb{R}^4: \bm x - \bm x_{lower}\leq \bm y\leq  \bm x + \bm x_{up}\}.$$
where $\bm x_{lower} = \bm x_{up} = (0.2, 0.2, 0.1, 0.1)$.

The simulated trajectories are shown in Fig. \ref{fig:traj}. The trajectory of the TLC method slightly enters the unsafe region due to the inter-sampling effect, while the trajectories from both the HOCBF and event-triggered TLC methods can ensure collision-free with respect to the obstacle.

 \begin{figure}[thpb]
	\centering
	\includegraphics[scale=0.5]{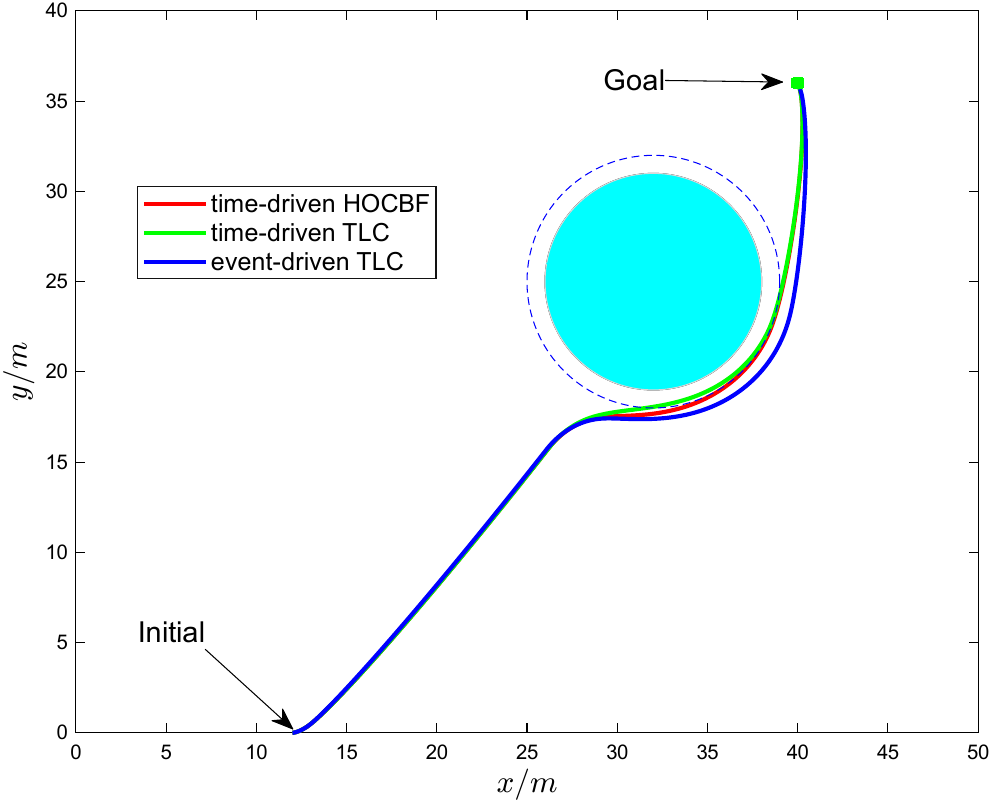}
	\caption{The simulated trajectories of the robot under the HOCBF, TLC, and event-triggered TLC methods.}	
	\label{fig:traj}
\end{figure}

The state, control and event time profiles are shown in Fig. \ref{fig:state}. They present similar patterns in terms of the controls and safety function $h(\bm x)$. As shown in the third frame of Fig. \ref{fig:state}, the $h(\bm x)$ of the TLC slightly goes below 0 due to the inter-sampling effect. Compared to the HOCBF method, the event-triggered TLC method can save computation resource as the number of QPs is much smaller. We further show the trajectories of $\psi_1(\bm x) = L_fh(\bm x) + \frac{1-1i}{\Delta t}h(\bm x)$ in TLC and event-triggered TLC when compared with the HOCBF method in Fig. \ref{fig:sets-robot}. This further validates that the TLC method extends the HOCBF method to the imaginary value domain.

\begin{figure}[thpb]
	\centering
	\includegraphics[scale=0.5]{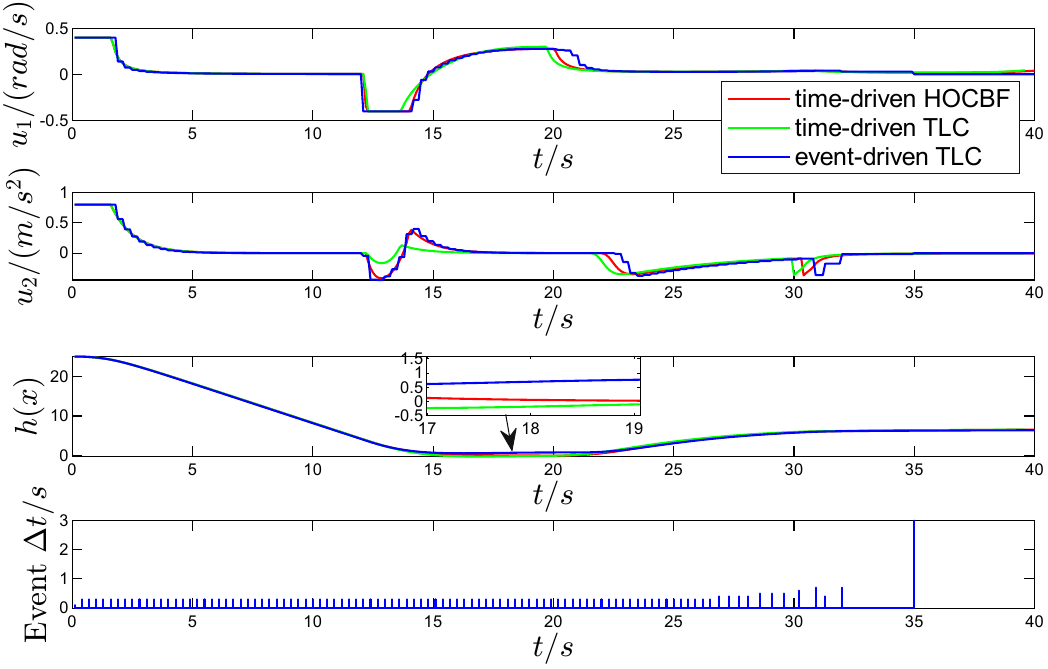}
	\caption{Comparison of State, control, safety function $h(\bm x)$, and event-triggering time $\Delta t$ between time-driven HOCBF and time-driven TLC, and event-triggered TLC in the robot control problem.}	
	\label{fig:state}
\end{figure}

\begin{figure}[thpb]
	\centering
	\includegraphics[scale=0.7]{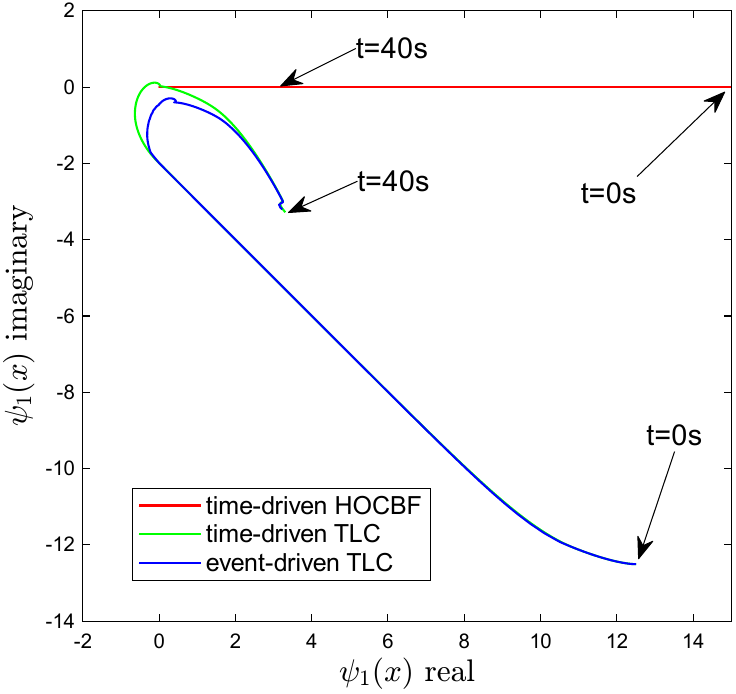}
	\caption{The variations of the corresponding $\psi_1(\bm x) = L_fh(\bm x) + \frac{1-1i}{\Delta t}h(\bm x)$ for TLC and event-triggered TLC when compared with the HOCBF in the complex plane.}	
	\label{fig:sets-robot}
\end{figure}

\section{Conclusion}
\label{sec:conclusion}

We proposed a new Taylor-Lagrange control method for ensuring the safety and stability of nonlinear control systems using Taylor's theorem with Lagrange remainder. The proposed Taylor-Lagrange control method pushes the boundary of existing control barrier function methods, and its existence is a necessary and sufficient condition fo system safety. We also proposed an event-triggered Taylor-Lagrange control method to address the inter-sampling issue in optimal control problems.  We applied the proposed method to an adaptive cruise control problem and to a 2D robot navigation task. 
In the future, we plan to study the receding horizon control of the Taylor-Lagrange control method as its current form allows the existence of controls in future time. We will also investigate its robustness and adaptivity for uncertain systems, as well as study higher-order Taylor-Lagrange control to address the inter-sampling effect.


%





\ifCLASSOPTIONcaptionsoff
  \newpage
\fi

\clearpage



%




\bibliographystyle{IEEEtran}
\bibliography{HOCBF}

%

\begin{IEEEbiography}[{\includegraphics[width=1in,height=1.25in,clip,keepaspectratio]{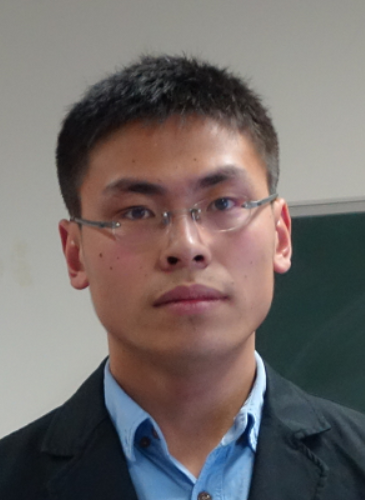}}]{Wei Xiao}
is currently an assistant professor at the robotics engineering department, Worcester Polytechnic Institute. He was a postdoctoral associate at Massachusetts Institute of Technology, and received a B.Sc. degree from the University of Science and Technology Beijing, China in 2013, a M.Sc. degree from the Chinese Academy of Sciences (Institute of Automation), China in 2016, and a Ph.D. degree from the Boston University, Brookline, MA, USA in 2021.
His research interests include control theory and machine learning, with particular emphasis on robotics and multi-agent control. He received an Outstanding Student Paper Award at the 2020 IEEE Conference on Decision and Control.
\end{IEEEbiography}

\begin{IEEEbiography}[{\includegraphics[width=1in,height=1.25in,clip,keepaspectratio]{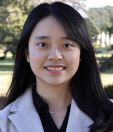}}]{Anni Li} is an assistant professor at the 
Electrical and Computer Engineering Department, The University of North Carolina at Charlotte. She received a B.Sc. degree in Mathematics and a minor in Physics from Central China Normal University, Wuhan, China, and an M.Sc. degree in operational research and cybernetics from Tongji University, Shanghai, China, and a Ph.D. degree in Systems Engineering at Boston University, Brookline, MA, USA in 2017, 2020 and 2025, respectively. Her research focuses on autonomous vehicles in transportation systems, with emphasis on safe and optimal cooperation and methods for cooperative compliance for social optimality.
\end{IEEEbiography}





\end{document}